\documentclass[letterpaper, 10 pt, conference]{ieeeconf} 
\usepackage{graphicx} 
\usepackage{amsmath} 

\usepackage{amsthm}
\usepackage{amssymb}
\usepackage{color}
\usepackage{url}

\IEEEoverridecommandlockouts                              
\newtheoremstyle{mytheoremstyle}
  {3pt}
  {3pt}
  {\normalfont}
  {}
  {\bfseries}
  {.}
  {.5em}
  {}

\theoremstyle{mytheoremstyle}

\newtheorem{mydef}{Definition}
\newtheorem{mythm}{Theorem}
\newtheorem{myprob}{Problem}
\newtheorem{mylem}{Lemma}

\newtheorem{mypro}{Proposition}

\newtheorem{remark}{Remark}

\title{\LARGE \bf Certificates Synthesis for A Class of Observational Properties in Stochastic Systems: A Unified Approach}

\author{Bohan Cui, Jianing Zhao, Yu Chen, Alessandro Abate, Marta Kwiatkowska and Xiang Yin
\thanks{This work was supported by  the National Natural Science Foundation of China (62573291,62533017,62173226).}
	\thanks{Bohan Cui, Jianing Zhao, Yu Chen and Xiang Yin are with the School of Automation and Sensing, Shanghai Jiao Tong University, Shanghai 200240, China.
	{\tt\small \{bohan\_cui, jnzhao, yuchen26, yinxiang\}@sjtu.edu.cn}.     
	Alessandro Abate and Marta Kwiatkowska are with the Department of Computer Science, University of Oxford, Oxford, OX1 3QD, UK.
	{\tt\small \{alessandro.abate, marta.kwiatkowska\}@cs.ox.ac.uk}. Corresponding Author: Xiang Yin.
	}}

\begin{document}

\maketitle

\begin{abstract}
In this paper, we investigate the probabilistic formal verification of stochastic dynamical systems over continuous state spaces. Motivated by problems in state estimation and information-flow security, we introduce the notion of \emph{observational properties}, which characterize the inferences an external observer can draw from system outputs. These properties are formulated as probabilistic hyperproperties based on HyperLTL over finite traces, yielding a unified framework that subsumes several existing notions studied separately in the literature.
We reduce the verification problem to reachability analysis over an augmented structure that integrates the system dynamics with an automaton representation of the specification. Building on this construction, we develop stochastic barrier certificates that provide probabilistic guarantees for property satisfaction while avoiding explicit state-space discretization.  
\end{abstract}

\section{Introduction}
Stochastic dynamical systems arise in a broad range of safety- and security-critical applications. In such settings, providing formal guarantees on system behavior under stochastic disturbances is essential \cite{lavaei2022automated,liu2022secure}.
In practice, however, these systems are often only \emph{partially observable}, either due to limited sensing or restricted information access. This raises a fundamental question: can an external observer infer sufficient knowledge about the system from available outputs? This question motivates the study of \emph{observational properties} \cite{hadjicostis2020estimation}. 
In the literature, a variety of observational properties have been extensively studied for partially observed systems. Representative observational properties include diagnosability and detectability, which concern information reconstruction \cite{sampath1995diagnosability,shu2007detectability}, and opacity, which concerns information-flow security \cite{lin2011opacity}. 

However, most of the above works are developed within the framework of Discrete-Event Systems (DES) and primarily focus on deterministic dynamics.
Although some extensions to stochastic DES have been proposed \cite{chen2025general}, they are largely restricted to discrete state spaces.
For systems operating over continuous domains, formal verification has been extensively studied for safety, reachability, and temporal-logic specifications \cite{prajna2004stochastic,jagtap2018temporal}. More recently, these techniques have been extended to certain observational properties, such as (approximate) opacity \cite{yin2020approximate,liu2024approximate}, diagnosability \cite{zhong2026verification}, and predictability \cite{de2020approximate,dong2024verification}. However, existing approaches are typically developed in a case-by-case manner, with verification structures and barrier conditions tailored to specific properties.

In this work, we study the formal verification of observational properties for stochastic dynamical systems. In contrast to existing approaches that develop property-specific, case-by-case methods, our goal is to provide a unified framework. 
Our approach is built upon HyperLTL$_f$, a finite-trace variant of HyperLTL \cite{clarkson2014temporal}, which extends linear-time temporal logic with trace quantification and naturally captures relations among observation-indistinguishable trajectories. Within this formalism, we characterize a class of observational properties, such as approximate detectability and approximate opacity, in a uniform manner.
Building on this representation, we reduce the verification problem to reachability analysis over an augmented structure that integrates the system dynamics with an automaton encoding of the specification. This construction enables a common verification procedure across different properties. Furthermore, we derive barrier-certificate-based conditions that provide probabilistic guarantees directly on continuous-state systems, thereby avoiding explicit state-space discretization.  Due to space constraints, all proofs are omitted; they are available, together with a case study, in \cite{cui2026certificates}.

Our work is related to \cite{zhao2024unified}, which proposes a HyperLTL-based framework for unifying observational properties of partially observed discrete-event systems. However, that framework is restricted to automata-theoretic techniques, making it inapplicable to stochastic systems   and continuous-state dynamics. In contrast, our work extends this line of research to partially observed stochastic systems and develops a verification framework that directly handles probabilistic dynamics in continuous domains.

\section{Preliminaries}\label{sec:preliminaries}

\textbf{Notations}: 
We denote by \(\mathbb{R},\mathbb{R}_{\geq0},\mathbb{N}, \mathbb{N}^+\) the set of all real numbers, non-negative real numbers, non-negative
integers, and positive integers, respectively. 
Given a Borel space \(A\), \(\mathcal{B}(A)\) and \(\mathcal{P}(A)\) represent its Borel \(\sigma\)-algebra and the set of Borel probability measures on \(A\), respectively. 
Consider a set \(X\), a string \(s=x_0x_1\cdots  x_n\in X^*\) is a finite sequence over \(X\),  
where $X^*$ denotes the set of all strings including the empty string $\epsilon$.
We denote by \(|s|\) the length of the string \(s\). For string \(s\in X^*\), we denote by \(s{\scriptstyle [i]}=x_i\) its \(i\)-th item, and by \(s{\scriptstyle [i,j]}=x_ix_{i+1}\cdots x_j\) its substring between \(i\)-th items and \(j\)-th items for \(0\leq i<j\leq |s|\).

\subsection{Partially Observable Stochastic Systems} 
In this paper, we consider a partially observable discrete-time stochastic system (POSS), which is a tuple 
$\mathcal{M}=(X,X_0,Y,W,\mu,f,\Omega,T)$, 
where 
 \(X\subseteq \mathbb{R}^{d_x}\) is the state space;
 \(X_0\subseteq X\) is the set of initial states;
  \(Y\subseteq \mathbb{R}^{d_y}\) is the output space;
 \(W \subseteq \mathbb{R}^{d_w}\) is the  disturbance space;
  \(\mu:\mathcal{B}(W)\to[0,1]\) is the probability measure of the disturbance;
   \(f: X\times W \to X\) is the system dynamics function;
   \(\Omega: X\to Y\) is the output mapping; and
 \(T\in\mathbb{N}\) is a positive integer representing the system horizon of our interest.
Note that the state space of a POSS is generally uncountably infinite when $X$ is continuous; thus, a POSS can be viewed as a continuous-state hidden Markov model. 

In this paper, we focus on autonomous systems without control inputs, as our primary interest is in the verification problem.  Specifically,  in $\mathcal{M}$, we are interested in the state trajectories within a fixed finite time horizon \(T\in\mathbb{N}\). A trajectory with horizon \(T\) in \(\mathcal{M}\) is a string of states 
$s=x_0x_1\cdots x_T\in X^{T+1}$, 
such that \(x_0\in X_0\), and for every \(0\leq  t\leq T\), we have \(x_{t+1}=f(x_t,w_t)\) for some disturbance \(w_t\in W\). 
The output of the state trajectory \(s\) is defined by
\[
\Omega(s)=\Omega(x_0)\Omega(x_1)\cdots \Omega(x_T)\in Y^{T+1}.
\]
We  denote by \(\mathcal{S}_{x_0,T}\subseteq X^{T+1}\) the set of all state trajectories with horizon \(T\) starting from \(x_0\in X_0\); 
and define \(\mathcal{S}_T=\bigcup_{x_0\in X_0}\mathcal{S}_{x_0,T}\) as the set of all state trajectories  with horizon \(T\) starting from an arbitrary initial state in \(\mathcal{M}\).

In this paper, we assume sets \(X,X_0, Y, W\), as well as functions \(f\) and \(\Omega\), are all Borel-measurable.  
In this paper, we consider the verification problem \emph{for all} possible initial states.
Given an initial state \(x_0\in X_0\), a system \(\mathcal{M}\) induces a discrete-time Markov process \(\{X_t\}_{t=0}^{T}\) over the state space \(X\) as follows: (1) \(X_0= x_0\); (2) \(X_{t+1}\sim \mathsf{K}(X_t,\cdot)\), 
where \(\mathsf{K}:X\times\mathcal{B}(X)\to[0,1]\) is the transition kernel such that,   \(\forall B\in\mathcal{B}(X),x\in X\), we have
$
\mathsf{K}(x,B)=\mu(\{w\in W: f(x,w)\in B\})$.
We denote by \(\mathbb{P}_{x_0}\in \mathcal{P}(\mathcal{S}_{x_0,T})\) the probability measure over the sample space  $\mathcal{S}_{x_0,T}$.  

\subsection{Observational Equivalence and State Estimates}
In the partial observation setting, the system behavior is inferred from the system dynamics together with the sequence of observed outputs.
We assume that the system model and the output trajectory are available to the observer.
Moreover, the observer is equipped with a finite measurement precision characterized by a parameter $\epsilon>0$, which induces indistinguishability between outputs that are within an $\epsilon$-neighborhood of each other.
This indistinguishability is captured by the following notion of observational equivalence.

\begin{mydef}[\bf \(\epsilon\)-Approximate Observational Equivalence]\upshape
Let \(\epsilon\in \mathbb{R}_{\geq 0}\) be a constant characterizing the measurement precision of the observer. 
We say two  state trajectories \(s,s'\in \mathcal{S}_T\) are \(\epsilon\)-approximately indistinguishable if 
\[
\forall i=0,\cdots ,T:  \|\Omega(s){\scriptstyle [i]}-\Omega(s'){\scriptstyle [i]}\|\leq\epsilon.
\]
For each trajectory \(s\in \mathcal{S}_T\), we define
\[
\hat{\mathcal{S}}_\epsilon(s)=\{s'\in \mathcal{S}_T:\max_i\|\Omega(s){\scriptstyle [i]}-\Omega(s'){\scriptstyle [i]}\|\leq\epsilon\}
\]
as the set of all \(\epsilon\)-approximately indistinguishable state trajectories of \(s\). 
For the sake of simplicity, we consider the $\mathcal{L}_2$-norm in this paper, 
but it can be replaced by any valid norm without affecting any result.
\end{mydef}

Intuitively, $\hat{\mathcal{S}}_\epsilon(s)$ characterizes the set of all state trajectories that are indistinguishable from the observer’s perspective when the true trajectory is 
$s$. Based on this set, one can further infer the possible current states of the system as well as the set of possible initial states from which the system may have originated; this is commonly referred to as the state estimation problem.
In this paper, we consider two types of state estimates. All observational properties studied in this work are defined based on these state estimates.

\begin{mydef}[\bf State Estimates]
Let \(s \in \mathcal{S}_T\) be a state trajectory of length \(T\) that is actually generated by the system. Then, we define:
    \begin{itemize}
        \item The \emph{initial-state estimate} (w.r.t.\  $\epsilon$) upon \(s\) is the set of initial states the system could start from initially, i.e.,
        \begin{equation}
            \hat{X}_0(s)=\{s'{\scriptstyle [0]}\in X_0:  s'\in \hat{\mathcal{S}}_\epsilon(s)\};
        \end{equation}
        \item The \emph{current-state estimate} (w.r.t.\ $\epsilon$) upon \(s\) is the set of states the system could be in currently, i.e.,
        \begin{equation}
            \hat{X}_T(s)=\{s'{\scriptstyle [T]}\in X:   s'\in \hat{\mathcal{S}}_\epsilon(s)\}.
        \end{equation}
    \end{itemize}
\end{mydef}

\section{A Unified Formulation of Observational Properties}\label{sec:obs}
In this section, we first revisit several classical notions of observational properties.
We then propose a unified formulation based on HyperLTL$_f$, which subsumes these properties within a common formal framework.
\subsection{Observational Properties}
In the context of formal verification of dynamical systems, including both discrete-event systems and continuous-state systems, observational properties concern the ability to infer certain information about the system based on available observations, i.e., the information flow or output trajectories in our setting.

Depending on the role of the observer, observational properties can be broadly categorized into two classes.
\begin{itemize}
    \item 
    The first class focuses on gaining sufficient knowledge about task-relevant  information. Representative notions in this category include detectability  \cite{shu2007detectability},  observability \cite{tong2022verification}, and diagnosability \cite{sampath1995diagnosability}.
    \item 
    The second class focuses on restricting information disclosure to ensure security and privacy. Representative notions in this category include opacity \cite{lin2011opacity,balun2021comparing}, anonymity \cite{yin2016uniform}, and non-interference \cite{ran2024noninterference,basile2025verification}.
\end{itemize}
Hereafter, we briefly review two representative observational properties, one from each category. Note that our definitions slightly differ from those in the existing literature, as we explicitly consider metric systems and incorporate observation precision into the formulation.

\subsubsection{\bf Detectability}
Detectability concerns the ability to uniquely infer the system state, either the current state or the initial state, from its outputs.
In continuous-state settings, exact state reconstruction is generally infeasible. 
To quantify approximate reconstruction, for a set \(A \subseteq X\) and a point \(x \in A\), we define the radius of \(A\) centered at \(x\) as
$\mathrm{rad}(A,x) := \sup_{y \in A} \|x-y\|$, which is the maximal distance from $x$ to the boundary of $A$.
We then introduce a tolerance parameter \(\lambda > 0\), requiring that the estimation error is bounded within a \(\lambda\)-neighborhood of the true state.
Furthermore, we consider a probabilistic notion, requiring that the probability of successful detection exceeds a threshold \(p \in (0,1]\).

\begin{mydef}[\bf Approximate Detectability]
Let \(\mathcal{M}\) be a POSS,
\(\epsilon \in \mathbb{R}_{\geq 0}\) denote the measurement precision, 
\(\lambda \in  \mathbb{R}_{\geq 0}\) denote the state detection tolerance, 
and \(p \in (0,1]\) denote the required detection probability.
We say that  \(\mathcal{M}\) is
\begin{itemize}
     \item 
     \((\epsilon,p,\lambda)\)-approximately initial-state detectable if, for each possible initial state   $x_0\in X_0$, we have
     \begin{equation}
     \mathbb{P}_{x_0,T}\left(\{s\in \mathcal{S}_{x_0,T}:\mathrm{rad}(\hat{X}_0(s),x_0) \le \lambda\}\right) \ge p;
     \end{equation}
    \item 
     \((\epsilon,p,\lambda)\)-approximately current-state detectable if, for each possible initial state   $x_0\in X_0$, we have
     \begin{equation}
     \mathbb{P}_{x_0,T}\left(\{s\in \mathcal{S}_{x_0,T}:\mathrm{rad}(\hat{X}_T(s),s{\scriptstyle [T]}) \le \lambda\}\right) \ge p.
     \end{equation}
\end{itemize}
\end{mydef}

Intuitively, detectability requires that, for any initial state and after \(T\) steps, the corresponding state estimate can be localized within at most \(\lambda\) with probability at least \(p\), under observation precision \(\epsilon\).

\subsubsection{\bf Opacity}
Opacity is an information-flow security property that characterizes whether certain secret information can be inferred by an external observer. 
In this setting, the observer is modeled as a passive intruder, and the system is associated with a set of secret states, denoted by \(X_S \subseteq X\).
Opacity requires that the intruder cannot determine, based on the observed outputs, whether the system is currently at or originated from the secret set.
Here, we introduce two opacity notions under the approximate and probabilistic setting \cite{liu2024approximate}.

\begin{mydef}[\bf Approximate Opacity]
Let \(\mathcal{M}\) be a POSS,
\(\epsilon \in \mathbb{R}_{\geq 0}\) denote the measurement precision, 
and \(p \in (0,1]\) denote the required secure probability. 
Let  \(X_S \subseteq X\) be a set of secret states. 
We say that  \(\mathcal{M}\) is
\begin{itemize}
     \item 
     \((\epsilon,p)\)-approximately initial-state opaque if, for each possible initial state   $x_0\in X_0$, we have
     \begin{equation}
     \mathbb{P}_{x_0}\left(\{s\in \mathcal{S}_{x_0,T}:\hat{X}_0(s)\not\subseteq X_S\}\right) \ge p;
     \end{equation}
    \item 
     \((\epsilon,p)\)-approximately current-state opaque if, for each possible initial state   $x_0\in X_0$, we have
     \begin{equation}
     \mathbb{P}_{x_0}\left(\{s\in \mathcal{S}_{x_0,T}:\hat{X}_T(s)\not\subseteq X_S\}\right) \ge p.
     \end{equation}
\end{itemize}
\end{mydef}

\begin{remark}
In the classical definition of current-state opacity for discrete-event systems, the observer is required to remain uncertain about whether the system is in a secret state at \emph{every time instant} along the trajectory, 
i.e., 
$\hat{X}_i(s{\scriptstyle [0,i]})\not\subseteq X_S,\forall 0\leq i \leq T$.
In contrast, our formulation evaluates current-state opacity only at the terminal time 
$T$. This restriction is introduced for technical convenience, allowing opacity to be naturally incorporated into our unified framework.
The extension to the standard, trajectory-wide notion is straightforward in principle, but entails additional modeling complexity, and is therefore omitted here for clarity.
\end{remark}

\subsection{A Unified  Formulation of Observational Properties}
Although individual certificates can be constructed for each of the observational properties discussed above, this approach generally requires bespoke verification structures tailored to each property. Examining the definitions of these properties, it becomes clear that they fundamentally involve reasoning over multiple trajectories, rather than a single trajectory as in the case of standard linear-time properties.
For instance, a trajectory satisfies the detectability requirement if \textbf{for all} trajectories that are $\epsilon$-closed in the output, the resulting states are $\lambda$-closed. Conversely, a trajectory satisfies the opacity requirement if \textbf{there exists} a trajectory that is $\epsilon$-closed in the output but leads to a non-secret state.

The existential or universal quantification over other trajectories can be formally captured within the framework of hyperproperties. Building on this insight, we propose a unified formulation of observational properties inspired by HyperLTL over finite traces (HyperLTL$_f$), leveraging the expressiveness of  hyperproperties while employing LTL$_f$ to specify temporal properties over finite traces.

To formulate the observational properties, we consider a \emph{trajectory property} that has the following syntax:
\begin{align}\label{syntax}
    &\phi::=\exists s'.\varphi \mid \forall s'. \varphi, \nonumber \\
    &\varphi::= \mu \mid \varphi\vee\varphi \mid \neg \varphi \mid \bigcirc\varphi \mid \varphi\, U\varphi,
\end{align}
where \(\exists\) and \(\forall\) are the existential and universal trace quantifiers, respectively. 
Symbols
\(\neg\) and \(\vee\) are Boolean operators ``negation" and ``disjunction", respectively; \(\bigcirc\) and \(U\) are temporal operators ``next" and ``until", respectively.
One can induce operators ``implication" by \(\varphi_1\to\varphi_2\equiv \neg \varphi_1\wedge\varphi_2\), ``conjunction" by \(\varphi_1\wedge\varphi_2\equiv\neg(\neg\varphi_1\vee\neg\varphi_2)\), ``eventually" by \(\Diamond\varphi\equiv\top U \varphi\), and ``always" by \(\square\varphi\equiv\neg\Diamond\neg\varphi\). 
Therefore, \(\varphi\) is a linear temporal logic formula defined over finite traces, but instead driven by the atomic predicate $\mu$, whose truth value is determined by a pair of states 
\[
\mu: X\times X \to\{\mathtt{T},\mathtt{F}\}.
\]
We denote by $AP$ the set of all atomic predicates  considered, and we define 
\[
L: X\times X\to 2^{AP}
\] 
as the labeling function that assigns the set of atomic predicates that holds true upon $(x,x')$, 
i.e., $L(x,x')=\{\mu\in AP: \mu(x,x')=\mathtt{T}\}$.
Given two state trajectories \(s,s' \in \mathcal{S}_T\) with same horizon $T$, we denote by 
\[
(s,s')=(s{\scriptstyle[0]}, s'{\scriptstyle[0]})(s{\scriptstyle[1]}, s'{\scriptstyle[1]})\cdots (s{\scriptstyle[T]}, s'{\scriptstyle[T]})
\]
the point-wise combination of \(s\) and \(s'\), 
and we define
\[
L(s,s')=L(s{\scriptstyle[0]}, s'{\scriptstyle[0]})L(s{\scriptstyle[1]}, s'{\scriptstyle[1]})\cdots 
L(s{\scriptstyle[T]}, s'{\scriptstyle[T]})
\]
the trace of atomic predicates along $(s,s')$.

We write \(s \models_{\mathcal{S}} \phi\) to denote that the trajectory \(s\) satisfies the formula \(\phi\) with respect to a global state trajectory set \(\mathcal{S}\).
Since in this work we focus on the verification problem, the trace set will always be the entire set of trajectories \(\mathcal{S}_T\) generated by \(\mathcal{M}\). Therefore, we omit the subscript and simply write \(s \models \phi\); its   semantics is then defined as follows:
\begin{align}
    &s\models\exists s'.\varphi &\text{iff} \quad &\exists s'\in \mathcal{S}_T: (s,s')\models\varphi \nonumber\\
    &s\models\forall s'.\varphi &\text{iff} \quad &\forall s'\in \mathcal{S}_T: (s,s')\models\varphi \nonumber\\
    &(s,s')\models \mu &\text{iff} \quad & \mu(s{\scriptstyle[0]},s'{\scriptstyle[0]})=\mathtt{T} \nonumber\\
    &(s,s')\models\varphi_1\vee\varphi_2 &\text{iff} \quad &(s,s')\models\varphi_1 \text{ or } (s,s')\models\varphi_2 \nonumber\\
    &(s,s')\models\neg\varphi &\text{iff} \quad &(s,s')\not\models\varphi \nonumber\\
    &(s,s')\models\bigcirc\varphi &\text{iff} \quad &(s,s'){\scriptstyle[1,T]}\models\varphi \nonumber\\
    &(s,s')\models\varphi_1 U\varphi_2 &\text{iff} \quad &\exists 0\leq i\leq T: \!(s,s'){\scriptstyle[i,T]}\models\varphi_2 \text{ and} \nonumber\\
    &&& \forall 0\leq m<i: (s,s'){\scriptstyle[m,T]}\models\varphi_1 \nonumber
\end{align}

In the above formulation, $s$ represents the state trajectory under evaluation, while $s'$ represents a quantified state trajectory from the set $\mathcal{S}_T$. The point-wise combination $(s, s')$ essentially constructs a trajectory pair where each element contains the state information of both trajectories at that time instant. This structure enables the atomic predicate $\mu$ to evaluate relational conditions, such as output indistinguishability, between the two runs at any given time instant, thus allowing the LTL formula $\varphi$ to reason about the joint evolution of these two traces.

The observational property verification problem considered in this work is defined as follows.

\begin{myprob}[\bf Observational Properties Verification]\upshape\label{def:property}
Given a POSS $\mathcal{M}$ and a trace property formula $\phi$, whose syntax is defined as in Equation~\eqref{syntax}, the goal is to determine whether, for any initial state $x_0 \in X_0$, the following holds:
\begin{equation}\label{eq:uni-def}
\mathbb{P}_{x_0}\big(\{s \in \mathcal{S}_{x_0,T} : s \models \phi\}\big) \geq p.
\end{equation} 
\end{myprob}

\subsection{Instantiations of the Unified Formulation}
Next, we illustrate how each specific property, such as detectability and opacity, can be expressed as a special instance within the unified formulation framework.

To this end, we first introduce some predicate notations. 
Given two states \(x,x'\in X\), 
we denote by
$x \overset{\epsilon}{\sim}_Y    x'$ the predicate that the outputs of two states are \(\epsilon\)-close, i.e., 
\[
x \overset{\epsilon}{\sim}_Y    x' =\mathtt{T}
\quad\text{iff}\quad
\|\Omega(x)-\Omega(x')\|\leq \epsilon.
\] 
Thus, for two trajectories $s,s'$,  formula 
$\square( s \overset{\epsilon}{\sim}_Y   s')$ represents that  they are \(\epsilon\)-approximately indistinguishable based on their outputs. 
We denote \( x \overset{\lambda}{\sim}_X   x'\) as the predicate 
that two  states  are  \(\lambda\)-close to each other in space $X$.

Now, we are ready to formulate detectability  within our unified framework.
\begin{mypro}
Under the probability requirement in Problem~1, approximate initial- and current-state detectability in Definition~3 are respectively characterized by
         \begin{equation}
        \phi_{id}=\forall s'.[\square( s \overset{\epsilon}{\sim}_Y    s')\to ( s \overset{\lambda}{\sim}_X    s')];
        \end{equation}
     \begin{equation}\label{eq:cd}
        \phi_{cd}=\forall s'.[\square( s \overset{\epsilon}{\sim}_Y    s')\to \Diamond\square( s \overset{\lambda}{\sim}_X    s')].
        \end{equation} 
\end{mypro}

Note that in Equation~\eqref{eq:cd}, since only finite traces are considered,
$\Diamond\square( s \overset{\lambda}{\sim}_X  s')$ 
is equivalent to $s{\scriptstyle[T]}\overset{\lambda}{\sim}_X  s'{\scriptstyle[T]}$. 
This  follows from the fact that, over finite horizons, the temporal operator $\Diamond\square$ reduces to a condition on the terminal time. In particular, it implies that all indistinguishable trajectories are $\lambda$-closed in the state space at the final time instant $T$.

To formulate opacity, we introduce the predicate
$\mathtt{Sec}:X\times X\to \{\mathtt{T},\mathtt{F}\}$
such that  
\[
\mathtt{Sec}( x, x')=\mathtt{T}\quad  \text{iff} \quad  x\in X_S. 
\]
For simplicity, we write $\mathtt{Sec}(x)$, since its value depends only on the first component. Similarly, we denote by $\mathtt{NS}(x,x')$ the predicate indicating that $x' \notin X_S$; we write it as $\mathtt{NS}(x')$ for simplicity. 
Now, we   formulate opacity  within our unified framework.

\begin{mypro}
Under the probability requirement in Problem~1, approximate initial- and current-state opacity in Definition~4 are respectively characterized by
    \begin{equation}\label{eq:io-cod}
        \phi_{io}=\exists s'.[\mathtt{Sec}(s)\to (\square( s \overset{\epsilon}{\sim}_Y    s')\wedge \mathtt{NS}(s'))];
        \end{equation}
        \begin{equation}\label{eq:co-cod}
        \phi_{co}=\exists s'.[\Diamond\square \mathtt{Sec}(s)\to (\square( s \overset{\epsilon}{\sim}_Y    s')\wedge \Diamond\square \mathtt{NS}(s'))].
        \end{equation} 
\end{mypro}
Intuitively, Equation~\eqref{eq:io-cod} states that, for any trajectory $s$, if it starts from a secret state, then there exists another trajectory $s'$ such that their outputs are $\epsilon$-close and $s'$ starts from a non-secret state. Equation~\eqref{eq:co-cod} is analogous; the only difference is that the secret status of the trajectories $s$ and $s'$ is evaluated at the final time instant via the operator $\Diamond\square$.

\section{Certificates for Observational Properties}\label{sec:certificates}
Since the state space of the system is continuous, the verification problem is, in general, undecidable. In this paper, we aim to derive sufficient conditions that can be effectively checked to guarantee the satisfaction probability by means of certificate synthesis.

Note that for a given LTL formula $\varphi$ defined over finite traces,
its acceptance can be characterized by the reachability property over a
deterministic finite automaton (DFA) $\mathcal{A}_\varphi$.
Specifically, for formula $\varphi$ defined over two trace variables, we construct a DFA
$
  \mathcal{A}_\varphi
  =(Q,q_0,2^{AP},\delta,Acc )$,
where $Q$ is a finite set of states, $q_0 \in Q$ is the initial state, $2^{AP}$ is the alphabet, $\delta: Q \times 2^{AP} \to Q$ is the transition function, and $Acc \subseteq Q$ is the set of accepting states.
The automaton $\mathcal{A}_\varphi$ accepts exactly all traces that satisfy $\varphi$, i.e., 
$\delta(q_0,L( s, s'))\in Acc $ iff $ ( s, s')\models\varphi$. 
Note that, such automaton is well-defined as the number of atomic predicates is finite even though the state spaces are continuous. 

Leveraging this automata-based representation, we introduce the following
verification structure, which serves as the domain on which certificates
are defined.

\begin{mydef}[\bf Verification Structure]\upshape
Given the POSS $\mathcal{M}=(X,X_0,Y,W,\mu,f,\Omega,T)$ and the DFA
$\mathcal{A}_\varphi$ defined above, the \emph{verification structure}
is
\[
  \mathcal{V}_\varphi
  =(V,V_0,Y^2,W^2,\mu,\delta^V,Acc_\varphi,T),
\]
where
\begin{itemize}
  \item $V=X\times X\times Q$ is the state space;
  \item $V_0= X_0\times X_0\times\{q_0\}$ is the set of initial states;
  \item $Y^2 = Y \times Y $ is the product output space;
  \item $W^2 = W \times W$ is the product disturbance input space;
  \item $\delta^V:V\times W^2\to V$ is the
        transition function defined by: for any $(x_1,x_2,q)\in V$ and $w_1,w_2\in W$, we have
        \begin{align}
         \delta^V&((x_1, x_2, q),w_1,w_2)=  \nonumber \\
         & (f(x_1,w_1),f(x_2,w_2),\delta(q,L(x_1,x_2));\nonumber
        \end{align}
  \item $Acc_\varphi=X\times X\times Acc$ is the set of accepting product states; and
  \item $\mu$ is a probability measure defined on $W$, which only characterizes the stochasticity on the disturbance $w_1$;
  \item $T$ is the system horizon of our interest.
\end{itemize}
\end{mydef}

In the above defined verification structure, two disturbance inputs $w_1$ and $w_2$ play different roles. Specifically, 
\begin{itemize}
    \item $w_1$ is \textbf{stochastic}, drawn
        according to the probability measure $\mu$ on $W$; while
    \item $w_2$
        is \textbf{non-deterministic}, treated as an adversarial choice
        made by an environment player without any probabilistic law.
\end{itemize}

To resolve the non-determinism in $\mathcal{V}_\varphi$, the environment player needs to determine the choice of $w_2$ at each instant. This choice does not represent a physical disturbance, but rather a construct to find worst-case trace pairs for property evaluation. Formally, a policy of the environment player is defined as a tuple 
$\pi = (\pi_{init}, \pi_v)$, where:
\begin{itemize}
\item $\pi_{init}: X_0\to X_0$ is the choice of the initial state of the second component;
\item $\pi_v : V\times W \times \{0,1,\cdots, T-1\}\to W$ is the decision rule for selecting $w_2$ at each state $v$ upon the disturbance $w_1$ at time instant $t\in [0,T-1]$.
\end{itemize}
Let $\Pi$ denote the set of all such policies. 
For a given initial state $x_0 \in X_0$ of the first component, any policy $\pi \in \Pi$ determines the initial state of the product system as $v_0 = (x_0, \pi_{init}(x_0), q_0) \in V_0$ and specifies the disturbance input $w_{2,t} = \pi_v(v_t,w_{1,t},t)$ at every time step $t\in [0,T-1]$. Together with the stochastic disturbance $w_1$, these choices induce a unique probability measure $\mathbb{P}_{x_0}^\pi\in\mathcal{P}(V^{T+1})$ on the space of all finite trajectories of $\mathcal{V}_\varphi$ with horizon $T$, for the prescribed initial state $x_0 \in X_0$.

The core idea connecting the original system to this verification structure is that reaching the accepting set $Acc_\varphi$ at the final time $T$ in $\mathcal{V}_\varphi$ is strictly equivalent to the generated trace pair satisfying the temporal logic formula $\varphi$. 
 
We define
$
\mathsf{TR}(Acc_\varphi)=\{v_0v_1\cdots v_T\mid v_T\in Acc_\varphi\}$
as the event that the state trajectory $v_0v_1\cdots v_T$ is accepted by $\mathcal{V}_\varphi$. 
The connection between $\mathcal{M}$ and $\mathcal{V}_\varphi$ is formally established as follows:
\begin{mythm}\upshape\label{thm:verification}    
Let $\mathcal{M}$ be a POSS and let $\phi$ be an observational property.  The verification problem is to determine whether
\[
\forall x_0\in X_0, \mathbb{P}_{x_0}\big(\{s \in \mathcal{S}_{x_0,T} : s \models \phi\}\big) \geq p.
\]
Then we have
    \begin{itemize}
        \item 
        if $\phi$ is of the form \(\forall s'.\varphi\), 
          then the above condition is satisfied if
        \[
        [\forall x_0\in X_0]\big(\inf_{\pi\in\Pi}\mathbb{P}^\pi_{x_0}(\mathsf{TR}(Acc_\varphi))\geq p\big);
        \]
        \item if $\phi$ is of the form \(\exists s'.\varphi\),
        then the above condition is satisfied if
        \[
        [\forall x_0\in X_0]\big(\sup_{\pi\in\Pi}\mathbb{P}^\pi_{x_0}(\mathsf{TR}(Acc_\varphi))\geq p\big)
        \]
    \end{itemize} 
\end{mythm}

\begin{proof}
    To prove Theorem~\ref{thm:verification}, we first note that the property $\varphi$ is evaluated over two trajectories $(s, s^{\prime})$, the automaton $\mathcal{A}_\varphi$ accepts  $L(s,s')$ if and only if $\delta(q_{0},L(s,s'))\in Acc$, which is strictly equivalent to $(s,s')\models\varphi$.
    Therefore, in the verification structure $\mathcal{V}_{\varphi}$, reaching the set $Acc_{\varphi}=X\times X\times Acc$ at time $T$ is equivalent to the generated trajectory pair satisfying $(s,s')\models\varphi$.
    
Case of $\forall$:
The proof follows a backward-induction argument. For a fixed
$x_{0}\in X_{0}$, define the terminal value function $
J_{T}(v):=\mathbf{1}_{Acc_{\varphi}}(v)$, $v\in V$,
and, recursively, for $t=T-1,T-2,\ldots,0$, define\vspace{-3pt}
\[
J_{t}(v)
:=
\int_{W}
\inf_{w_{2}\in W}
J_{t+1}\bigl(\delta^{V}(v,w_{1},w_{2})\bigr)
\mu(\mathrm{d}w_{1}).\vspace{-3pt}
\]
Here, $J_t(v)$ represents the minimum probability of reaching
$Acc_{\varphi}$ at time $T$, starting from the product state
$v$ at time $t$, when the environment player selects the
second disturbance adversarially.
The policy $\pi^{*}=(\pi_{\mathrm{init}}^{*},
\pi_{v}^{*})$ is constructed backward in time as follows.
At every $t\in\{0,\ldots,T-1\}$, after observing the current
product state $v$ and the disturbance $w_{1}$, select\vspace{-3pt}
\[
\pi_{v}^{*}(v,w_{1},t)
\in
\arg\min_{w_{2}\in W}
J_{t+1}\bigl(\delta^{V}(v,w_{1},w_{2})\bigr).\vspace{-3pt}
\]
The initial state of the second component is selected according to $
\pi_{\mathrm{init}}^{*}(x_{0})
\in
\arg\min_{x_{0}'\in X_{0}}
J_{0}\bigl((x_{0},x_{0}',q_{0})\bigr)$.
We next relate the policy $\pi^{*}$ to the universal trajectory
property. Suppose that
$s\not\models\forall s'.\varphi$.
Then there exists a trajectory $s'\in\mathcal{S}_{T}$ such that
$
(s,s')\not\models\varphi$.
Equivalently, the corresponding product trajectory reaches a
state in $V\setminus Acc_{\varphi}$ at time $T$. Therefore, the backward-induction
adversarial policy $\pi^*$ can select, at each time step, a disturbance $w_{2,t}$ that
preserves such a violating continuation whenever one exists.
Consequently, the second component generated by $\pi^{*}$ follows
a violating trajectory, and hence 
$
s\not\models\forall s'.\varphi
\implies
\neg\mathsf{TR}(Acc_{\varphi})
\text{ under }\pi^{*}$.
Therefore,\vspace{-3pt}
\[
\left\{
s:
\mathsf{TR}(Acc_{\varphi})
\text{ under }\pi^{*}
\right\}
\subseteq
\left\{
s:
s\models\forall s'.\varphi
\right\}.\vspace{-3pt}
\]
Since the marginal distribution of the first trajectory $s$
under $\mathbb{P}^{\pi^{*}}_{x_{0},T}$ is
$\mathbb{P}_{x_{0},T}$, it follows that \vspace{-3pt}
\[
\begin{aligned}
p
&\leq
\inf_{\pi\in\Pi}
\mathbb{P}^{\pi}_{x_{0},T}
\bigl(\mathsf{TR}(Acc_{\varphi})\bigr)\\
&\leq
\mathbb{P}^{\pi^{*}}_{x_{0},T}
\bigl(\mathsf{TR}(Acc_{\varphi})\bigr)\\
&\leq
\mathbb{P}_{x_{0},T}
\bigl(s\models\forall s'.\varphi\bigr),
\end{aligned}\vspace{-3pt}
\]
which concludes the proof.

    Case of $\exists$: For any fixed policy $\pi \in \Pi$ we claim the following holds:
    \[
    \{s\mid \mathsf{TR}(Acc_\varphi) \text{ under } \pi\}  \subseteq \{s\mid s \models \exists s'. \varphi\}.
    \]
    This holds because if, for the $\pi$ induced trace pair $(s, s'_\pi(s))$, $(s, s'_\pi(s)) \models \varphi$, then $s'_\pi(s) \in \mathcal{S}_T$ serves as an explicit witness, and by the existential semantics, $s \models \exists s'. \varphi$ follows immediately. Since the marginal distribution of $s$ under $\mathbb{P}^\pi_{x_0,T}$ is $\mathbb{P}_{x_0,T}$, this inclusion gives:\vspace{-3pt}
    \[
    \mathbb{P}^\pi_{x_0,T}(\mathsf{TR}(Acc_\varphi)) \le \mathbb{P}_{x_0,T}(s \models \exists s'.\varphi).\vspace{-3pt}
    \]
    Since the right-hand side does not depend on $\pi$, taking the supremum over the left-hand side yields:\vspace{-3pt}
    \[
    \sup_{\pi \in \Pi} \mathbb{P}^\pi_{x_0,T}(\mathsf{TR}(Acc_\varphi)) \le \mathbb{P}_{x_0,T}(s \models \exists s'.\varphi).\vspace{-3pt}
    \]
    Therefore, if $\sup_{\pi \in \Pi} \mathbb{P}^\pi_{x_0,T}(\mathsf{TR}(Acc_\varphi)) \ge p$ holds for all $x_0 \in X_0$, we have $[\forall x_0 \in X_0](\mathbb{P}_{x_0,T}(s \models \varphi) \ge p)$.
\end{proof}

Based on Theorem~\ref{thm:verification}, the verification of observational properties in $\mathcal{M}$ is reduced to analyzing the reachability probability of $Acc_\varphi$ at time instant $T$ in the verification structure $\mathcal{V}_\varphi$. However, directly computing the reachability probability using traditional discrete model checking remains intractable since the state space $V$ is continuous. To address this, we adapt the stochastic barrier certificate framework to handle the mixed stochastic and non-deterministic nature of $\mathcal{V}_\varphi$, as stated formally in the following two lemmas.

\begin{mylem}
    Given a target set $V_G\subseteq V$, a function $\mathbf{V}:V\times\mathbb{N}\to\mathbb{R}$ is said to be a universal terminal reach barrier certificate ($\forall$-TRBC) for $V_G$ if the following conditions hold for some $\alpha>0$ and $\beta\in\mathbb{R}$:
    \begin{itemize}
        \item $\forall v\in V, \mathbf{V}(v,T)\leq \mathbf{1}_{V_G}(v)$;
        \item $\forall v\in V, t\in [1,T]$, 
        \[\int_{W}\!\inf_{w_2\in W}\!\!\mathbf{V}(\delta^V\!(v\!,w_1\!,w_2)\!,t)\mu(dw_1)\!\geq\! \frac{\mathbf{V}(v,t-1)}{\alpha}\!+\!\beta;
        \]
    \end{itemize}
    If $\mathbf{V}$ is a $\forall$-TRBC, then it holds that
    \[
    \begin{aligned}        \inf_{\pi\in\Pi}\mathbb{P}^\pi_{x_0,T}(&\mathsf{TR}(V_G))\geq \\ 
    &\inf_{x_0'\in X_0}\mathbf{V}((x_0,x_0',q_0),0)\alpha^{-T}+\big(\sum_{i=0}^{T-1}\!\!\alpha^{-i}\big)\beta.
    \end{aligned}
    \]
\end{mylem}

The intuition behind Lemma~1 is rooted in the dynamic programming principle, where value functions can be defined to characterize the minimum probability of reaching the target set. The barrier certificate $\mathbf{V}$ serves as a sequence of lower bounds on these optimal value functions at each time step via backward recursion.
Specifically, by constructing an auxiliary sequence $\eta_t(v)$ based on $\mathbf{V}$, one can show that the dynamic programming value function satisfies $u_t^{\forall}(v) \geq \eta_t(v)$ for all $t$.

While Lemma~1 establishes a lower bound for universal properties by taking the infimum over nondeterministic environmental choices, existential properties require evaluating the supremum reachability probability.  Lemma~2 provides a corresponding probabilistic bound for existential properties.

\begin{mylem}
    Given a target set $V_G\subseteq V$, a function $\mathbf{V}:V\times\mathbb{N}\to\mathbb{R}$ is said to be an existential terminal reach barrier certificate ($\exists$-TRBC) for $V_G$ if the following conditions hold for some $\alpha>0$ and $\beta\in\mathbb{R}$:
    \begin{itemize}
        \item $\forall v\in V, \mathbf{V}(v,T)\leq \mathbf{1}_{V_G}(v)$;
        \item $\forall v\in V, t\in [1,T]$,  
        \[\int_{W}\!\sup_{w_2\in W}\!\!\mathbf{V}(\delta^V\!(v\!,w_1\!,w_2)\!,t)\mu(dw_1)\!\geq\! \frac{\mathbf{V}(v,t-1)}{\alpha}\!+\!\beta;
        \]
    \end{itemize}
    If $\mathbf{V}$ is an $\exists$-TRBC, then it holds that
    \[
    \begin{aligned}        \sup_{\pi\in\Pi}\mathbb{P}^\pi_{x_0,T}(&\mathsf{TR}(V_G))\geq \\
    &\sup_{x_0'\in X_0}\mathbf{V}((x_0,x_0',q_0),0)\alpha^{-T}+\big(\sum_{i=0}^{T-1}\!\!\alpha^{-i}\big)\beta.
    \end{aligned}
    \]
\end{mylem}

The only difference between Lemma~1 and Lemma~2 lies in the treatment of non-determinism within the verification structure. While $\forall$-TRBC uses an infimum to bound the worst-case probability for universal properties, $\exists$-TRBC employs a supremum to provide a lower bound for existential properties, such as opacity. 
By combining the theoretical reduction from Theorem~1 with the barrier certificate bounds from Lemma~1 and Lemma~2, we can directly verify the overarching observational properties as follows.

\begin{mythm}
    Given the POSS \(\mathcal{M}\) and an observational property defined as $\forall x_0\in X_0, \mathbb{P}_{x_0}\big(\{s \in \mathcal{S}_{x_0,T} : s \models \phi\}\big) \geq p$. 
    \begin{itemize}
        \item if $\phi$ is of the form \(\forall s'.\varphi\), then $\mathcal{M}$ satisfies the observational property if there exists a $\forall$-TRBC $\mathbf{V}:V\times\mathbb{N}\to\mathbb{R}$ for $Acc_\varphi$ such that\vspace{-10pt}
        \[
        \inf_{x_0'\in X_0}\!\!\!\!\mathbf{V}((x_0,x_0',q_0),0)\alpha^{-T}+\big(\sum_{i=0}^{T-1}\!\!\alpha^{-i}\big)\beta\ge p;
        \]
        \item if $\phi$ is of the form \(\exists s'.\varphi\), then $\mathcal{M}$ satisfies the observational property if there exists an $\exists$-TRBC $\mathbf{V}:V\times\mathbb{N}\to\mathbb{R}$ for $Acc_\varphi$ such that\vspace{-10pt}
        \[
        \sup_{x_0'\in X_0}\!\!\!\!\mathbf{V}((x_0,x_0',q_0),0)\alpha^{-T}+\big(\sum_{i=0}^{T-1}\!\!\alpha^{-i}\big)\beta\ge p.
        \]
    \end{itemize} 
\end{mythm}

\begin{remark}
    While we have established the theoretical foundations for probability bounds, the practical construction of such certificates remains a significant challenge. 
    In the literature, several computational techniques have been developed to search for barrier certificates. 
    Traditional approaches primarily include Sum-of-Squares (SOS) programming \cite{prajna2004stochastic} or Satisfiability Modulo Theories (SMT) based methods \cite{abate2025quantitative}. Although these methods offer formal guarantees for polynomial or piecewise-linear systems, they scale poorly with the system's degree and dimension and often struggle with the complex stochastic expectations and non-linearities, which are inherent in the verification structure $\mathcal{V}_\varphi$. 
    Recently, the emergence of neural certificates \cite{rickard2026data} has provided a promising alternative by parameterizing the certificate as a deep neural network. Unlike symbolic methods, neural-based approaches can approximate highly non-linear functions without structural constraints on the plant dynamics.  
\end{remark}

\section{Case Studies}\label{sec:case}

In this section, we demonstrate the effectiveness of our proposed certificates by considering the following example. 

\textbf{Example Settings:} We consider the following discrete-time system\vspace{-6pt}
\[
x_{t+1}=f(x_t,w_t)=0.9x_t+w_t, 
\]
where $x_t\in X=\mathbb{R}$, $x_0\in X_0=[-2,2]$, the stochastic disturbance $w_t$ is sampled from a normal distribution $\mathcal{N}(0, 0.4^2)$ but is clipped within $[-1.2,1.2]$, the output mapping is defined as $\Omega(x)=x^2$ and time horizon $T=10$.
The property that we consider is $(\epsilon, p, \lambda)$-approximate detectability, where we set $\epsilon=0.5$ and $\lambda=0.8$. Specifically, we require that
\[
\mathbb{P}_{x_0,T}\left(\{s\in \mathcal{S}_{x_0,T}:\mathrm{rad}(\hat{X}_T(s),s{\scriptstyle[T]}) \le 0.8\}\right) \ge p.
\]
The Monte Carlo empirical probability of this property is obtained as $p=0.081$. 

\textbf{Certificate Construction:}
We utilize a deep neural network to approximate the certificate $\mathbf{V}(x, x', q, t)$. For the sake of simplicity, we fix the parameter $\alpha=1$ and optimize $\beta$. The model takes the state pair $(x, x')$, the embeddings of the discrete automaton state $q$, and the time step $t$ as input. The main multilayer perceptron (MLP) consists of three hidden layers with dimensions $[64, 64, 32]$, utilizing LeakyReLU activations.
The synthesis process is then formulated as an optimization problem where the loss function $\mathcal{L}_{total}$ encodes the certificate constraints. The total loss is defined as: 
\[
\mathcal{L}_{total} = \lambda_{term} \frac{\mathcal{L}_{term}}{\text{sg}(|\mathcal{L}_{term}|)} + \lambda_{rec} \frac{\mathcal{L}_{rec}}{\text{sg}(|\mathcal{L}_{rec}|)} + \lambda_{\beta} \frac{\mathcal{L}_{\beta}}{\text{sg}(|\mathcal{L}_{\beta}|)}
\]
where $\text{sg}(\cdot)$ denotes the stop-gradient operation, and weight coefficients are set to $\lambda_{term}=6.5, \lambda_{rec}=8.5, \lambda_{\beta}=2.5$. The individual components are defined as follows: 
\begin{itemize}
    \item The terminal constraint loss is defined as
    \[
    \mathcal{L}_{term} = \mathbb{E} \left[ 
    \begin{aligned}
        \sum_{v \in Acc_\varphi} \text{ReLU}(\mathbf{V}(v, T) - 1)\\
        +\sum_{v \not\in Acc_\varphi} \text{ReLU}(\mathbf{V}(v, T)) 
    \end{aligned} \right];
    \]
    \item The recursion loss is defined as:
    \[
    \mathcal{L}_{rec} = \mathbb{E} 
    \left[ 
    \begin{aligned}
        \text{ReLU} ( ( \frac{\mathbf{V}(v,t-1)}{\alpha} + \beta )\\
         - \int_{W}\min_{i=1 \dots N}\mathbf{V}(\delta^V(v,w_1,w_2),t)\mu(dw_1) ) ) 
    \end{aligned}  
    \right],
    \]
    where the infimum is approximated using samples;
    \item To avoid the trivial solution, we encourage $\beta > 0$:
    \[
    \mathcal{L}_{\beta} = \text{Softplus}(-\beta) - 0.01 \cdot \text{ReLU}(\beta).
    \]
\end{itemize}

\textbf{Numerical Results:}
    \begin{figure}
        \centering
        \includegraphics[width=\linewidth]{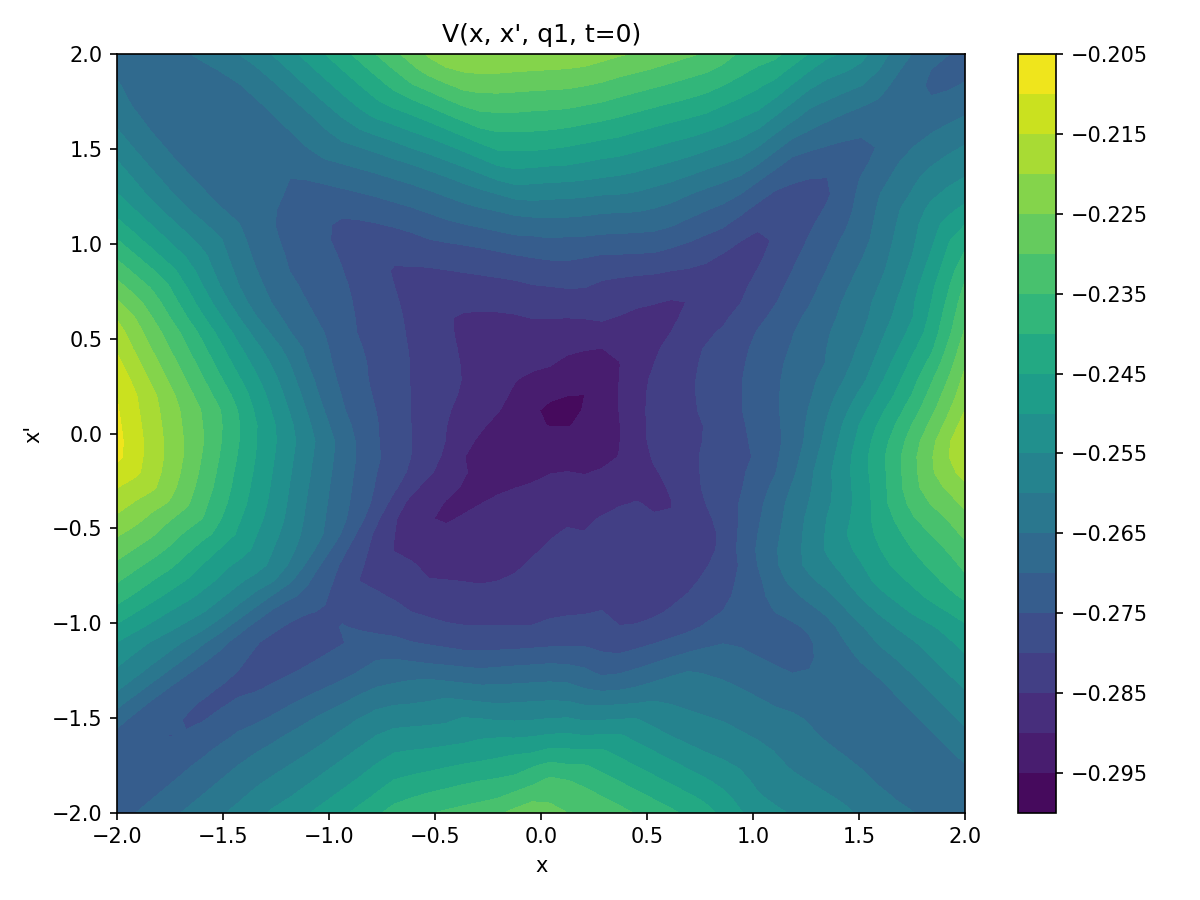}
        \caption{The value of the certificate $\mathbf{V}$ at time instant $t=0$.}
        \label{fig:value}
    \end{figure}
    The network was trained for $2000$ epochs on $100,000$ trajectory pairs. The optimization procedure converged successfully and produced a valid certificate with the learned parameter $\beta=0.0617$, yielding the probability guarantee $p\ge 0.0409$, the value of $\mathbf{V}$ at time instant $t=0$ is as shown in Figure~\ref{fig:value}. 

\textbf{Certificate Validation:}
    We validate the neural certificate via a dense grid sampling method after training. Specifically, we discretize the state space into a uniform grid with $50$ samples per dimension. For each of them, we evaluate the two constraints of the $\forall$-TRBC for each time step $t\in[0,10]$; the total violations of the constraints are $0$.
    
\section{Conclusion}\label{sec:conclusion}
In this paper, we propose a certificate-based framework for verifying probabilistic observational properties in partially observed stochastic systems. The framework is unified through a hyperproperty-based formulation that subsumes a range of existing notions. Moreover, it enables probabilistic guarantees to be established directly on continuous-state systems without explicit state-space discretization. Future work will focus on extending the framework to control synthesis under partial observation.

\bibliographystyle{plain}
\bibliography{references}

\appendix
\textbf{Proof of Proposition~1} 

\begin{proof}

    1) We prove the condition for \((\epsilon,p,\lambda)\)-approximate initial-state detectability by proving that $s \models \phi_{id} \Leftrightarrow \mathrm{rad}(\hat{X}_0(s),x_0) \le \lambda$.

    $(\Rightarrow)$: Suppose by contraposition that $\mathrm{rad}(\hat{X}_0(s),x_0)> \lambda$. This indicates there is a trajectory $s' \in \hat{\mathcal{S}}_\epsilon(s)$ where the initial states diverge beyond the threshold, i.e., $\|s[0] - s'[0]\| > \lambda$. Therefore, the observation equivalence $\square( s \overset{\epsilon}{\sim}_Y   s')$ holds. However, the state equivalence $ s \overset{\lambda}{\sim}_X   s'$ evaluated at the initial instant evaluates to false. This proves that $s \not\models \phi_{id}$.

    $(\Leftarrow)$: Suppose by contraposition that $s \not\models \phi_{id}$. This means there exists a trajectory $s'$ where the observations are globally $\epsilon$-close, $\square( s \overset{\epsilon}{\sim}_Y   s')$, but the initial states violate $s \overset{\lambda}{\sim}_X   s'$. Therefore, we have $s' \in \hat{\mathcal{S}}_\epsilon(s)$ but $\|s[0] - s'[0]\| > \lambda$. Thus, $\mathrm{rad}(\hat{X}_0(s),x_0) > \lambda$.

    2) We then prove the condition for \((\epsilon,p,\lambda)\)-approximate current-state detectability by proving that $s \models \phi_{cd} \Leftrightarrow \mathrm{rad}(\hat{X}_T(s),s{\scriptstyle [T]}) \le \lambda$.
    
    $(\Rightarrow)$: Suppose, for the sake of contraposition, that $\mathrm{rad}(\hat{X}_T(s),s{\scriptstyle [T]}) > \lambda$. This means there exists a trajectory $s' \in \hat{\mathcal{S}}_\epsilon(s)$ such that $\|s[T] - s'[T]\| > \lambda$. Since $s'$ and $s$ are $\epsilon$-approximately undistinguishable, the condition $\square(s \overset{\epsilon}{\sim}_Y   s')$ holds over the trajectories. However, since $\|s[T] - s'[T]\| > \lambda$, the state proposition $s \overset{\lambda}{\sim}_X   s'$ is violated at the final instant $T$. Consequently, the temporal property $\Diamond \square$ cannot hold. It follows that $s \not\models \phi_{cd}$, completing this direction.
    
    $(\Leftarrow)$: We still prove this direction by contraposition. Suppose that $s \not\models \phi_{cd}$.This implies there exists a trajectory $s'$ such that $\square(s \overset{\epsilon}{\sim}_Y   s')$ evaluates to true, but $\Diamond \square(s \overset{\lambda}{\sim}_X   s')$ evaluates to false. The truth of the premise $\square(s \overset{\epsilon}{\sim}_Y   s')$ implies that $s' \in \hat{\mathcal{S}}_\epsilon(s)$. The violation of $\Diamond \square(s \overset{\lambda}{\sim}_X   s')$ evaluated over the finite horizon means the states are not $\lambda$-close at the terminal instant $T$, resulting in $\|s[T] - s'[T]\| > \lambda$. Therefore, the trajectory $s$ fails the detectability condition, meaning $\mathrm{rad}(\hat{X}_T(s),s{\scriptstyle [T]}) > \lambda$.
\end{proof}

\textbf{Proof of Proposition~2} 

\begin{proof}
    1) We prove by equivalently proving that $s \models \phi_{io} \Leftrightarrow \hat{X}_0(s)\not\subseteq X_S$.

    $(\Rightarrow)$: Suppose by contraposition that $\hat{X}_0(s)\subseteq X_S$. This means that the initial-state estimate is entirely contained within the secret states. Because the true trajectory $s$ trivially belongs to its own $\epsilon$-approximation space ($s \in \hat{\mathcal{S}}_\epsilon(s)$), it must be that $s[0]\in X_S$, which means the proposition $\mathtt{Sec}(s)$ holds. Furthermore, for any trajectory $s' \in \hat{\mathcal{S}}_\epsilon(s)$, its initial state is also contained in the estimate, so $s'[0] \in X_S$. Consequently, for any alternative trajectory $s'$ where $\square(s \overset{\epsilon}{\sim}_Y   s')$ holds, the non-secret proposition $\mathtt{NS}(s')$ evaluates to false. Because the premise $\mathtt{Sec}(s)$ is true but the right side of the implication is always false, there exists no $s'$ satisfying the formula, meaning $x \not\models \phi_{io}$.

    $(\Leftarrow)$: Suppose by contraposition that $s \not\models \phi_{io}$. Because $\phi_{io}$ is governed by an existential quantifier, its negation requires that for all trajectories $s'$, the implication $\mathtt{Sec}(s)\to (\square( s \overset{\epsilon}{\sim}_Y    s')\wedge \mathtt{NS}(s'))$ evaluates to false. The only way an implication is uniformly false is if the premise is true and the conclusion is false. Therefore, $\mathtt{Sec}(s)$ must be true (meaning $s[0] \in X_S$). Also, for any $s'$ where $\square(s \overset{\epsilon}{\sim}_Y    s')$ holds, $\mathtt{NS}(s')$ must be false, meaning $s'[0] \in X_S$. Since every $s'$ that is $\epsilon$-indistinguishable from $s$ initiates in a secret state, the estimate resolves to $\hat{X}_0(s) \subseteq X_S$.

    2) We prove equivalently that $ s \models \phi_{co} \Leftrightarrow\hat{X}_T(s)\not\subseteq X_S$.

    $(\Rightarrow)$: Suppose by contraposition that the trajectory condition fails, meaning $\hat{X}_T(s) \subseteq X_S$. This implies the actual trajectory ends in a secret state, so $s[T] \in X_S$. Thus, the proposition $\Diamond \square \mathtt{Sec}(s)$  evaluates to true. Moreover, for every indistinguishable finite trajectory $s' \in \hat{\mathcal{S}}_\epsilon(s)$, it must hold that $s'[T] \in X_S$. Consequently, for any trace $s'$ matching the observation $\square(s \overset{\epsilon}{\sim}_Y    s')$, the corresponding state proposition $\Diamond \square \mathtt{NS}(s')$ must be false. Since the premise holds and the conclusion is consistently false for all possible $s'$, the existential trace property is violated, giving $s \not\models \phi_{co}$.

    $(\Leftarrow)$: Suppose by contraposition that $s \not\models \phi_{co}$. The negation of this existential formula implies that the premise $\Diamond \square \mathtt{Sec}(s)$ is true (hence $s[T] \in X_S$), and for any trajectory $s'$ satisfying the observation condition $\square(s \overset{\epsilon}{\sim}_Y    s')$, the non-secret proposition $\Diamond \square \mathtt{NS}(s')$ evaluates to false. This means every trajectory $s'$ that is observationally equivalent to $s$ has its current state $s'[T] \in X_S$. Therefore,  $\hat{X}_T(s) \subseteq X_S$, which means the trajectory prefix violates the requirement for current-state opacity.
\end{proof}

\textbf{Proof of Lemma~1}

\begin{proof}
    We prove from a dynamic programming perspective. Let $u_t^\forall: V \to \mathbb{R}$ for $t=0,1,\dots,T$ be the dynamic programming value functions representing the infimum probability of reaching $V_G$ at the terminal time $T$ starting from state $v$ at time $t$. 
    According to Theorem 11 in \cite{summers2010verification}, the sequence can be defined via the backward recursion: (i) $u_T^\forall(v) = \mathbf{1}_{V_G}(v)$, and (ii) for $t=0,\dots,T-1$,\vspace{-3pt}
    \[
            u_t^\forall(v) = \int_{W}\inf_{w_2\in W}u_{t+1}^\forall(\delta^V(v,w_1,w_2))\mu(dw_1),\vspace{-3pt}
    \]
    where $\mathbf{1}_{V'}:V\to \{0,1\}$ denotes the indicator of a set $V'\subseteq V$, namely, $
    \mathbf{1}_{V'}(v)=1\Leftrightarrow v\in V'$.
    Introduce the auxiliary functions \vspace{-3pt}
    \[
    \eta_t(v) = \mathbf{V}(v,t)\alpha^{t-T} + \big(\sum_{i=0}^{T-t-1}\alpha^{-i}\big)\beta.\vspace{-3pt}
    \]
    We show by backward induction that
    $u_t^\forall(v) \geq \eta_t(v)$ for every $v\in V$ and every $t=0,1,\dots,T$.

    \textbf{Induction Basis ($t=T$):}
    The terminal condition of the certificate gives \vspace{-3pt}
    \[
    u_T^\forall(v) = \mathbf{1}_{V_G}(v) \geq \mathbf{V}(v,T) = \mathbf{V}(v,T)\alpha^0 + 0 = \eta_T(v).\vspace{-3pt}
    \]

    \textbf{Induction Step:}
    Assume that $u_{t+1}^\forall(v) \geq \eta_{t+1}(v)$ for some $0 \leq t < T$. Using the monotonicity of integration and of the infimum operator, we obtain, for every $v\in V$,\vspace{-3pt}
    \[
    \begin{aligned}
        &u_t^\forall(v)= \int_{W}\inf_{w_2\in W}u_{t+1}^\forall(\delta^V(v,w_1,w_2))\mu(dw_1)\geq \\
        &\int_{W}\inf_{w_2\in W} \eta_{t+1}(\delta^V(v,w_1,w_2)) \mu(dw_1)= \big(\!\!\!\!\sum_{i=0}^{T-t-2}\!\!\!\!\alpha^{-i}\big)\beta+\\
        &\alpha^{t+1-T}\!\! \left(\! \int_{W}\!\!\inf_{w_2\in W}\!\!\mathbf{V}(\delta^V(v\!,w_1\!,w_2),t
        \!+\!1)\mu(dw_1)\! \right).
    \end{aligned}\vspace{-3pt}
    \]
    The second certificate condition then implies \vspace{-3pt}
    \[
    \begin{aligned}
        u_t^\forall(v) &\geq \alpha^{t+1-T} \left( \frac{\mathbf{V}(v,t)}{\alpha} + \beta \right) + \big(\sum_{i=0}^{T-t-2}\alpha^{-i}\big)\beta \\
        &= \mathbf{V}(v,t)\alpha^{t-T} + \alpha^{t+1-T}\beta + \big(\sum_{i=0}^{T-t-2}\alpha^{-i}\big)\beta \\
        &= \mathbf{V}(v,t)\alpha^{t-T} + \big(\sum_{i=0}^{T-t-1}\alpha^{-i}\big)\beta = \eta_t(v).
    \end{aligned}\vspace{-3pt}
    \]
    Thus, the desired comparison holds at time $t$, completing the backward induction.
    In particular, at $t=0$, we have $u_0^\forall(v) \geq \eta_0(v)$. Since the policy-induced nondeterminism also includes the choice of the initial state $x_0'\in X_0$ of the second trajectory, the overall worst-case probability is obtained by minimizing over $x_0'$. Therefore,
    \[
    \begin{aligned}
        \inf_{\pi\in\Pi}\mathbb{P}^\pi_{x_0}(\mathsf{TR}(V_G)) &= \inf_{x_0'\in X_0} u_0^\forall((x_0,x_0',q_0))\\ &\geq \inf_{x_0'\in X_0} \eta_0((x_0,x_0',q_0)).
    \end{aligned}
    \]
    Substituting the definition of $\eta_0$ yields the claimed lower bound.
\end{proof}

\textbf{Proof of Lemma~2}

\begin{proof}
    The proof follows a symmetric structure similar to that of the $\forall$-TRBC by changing $\inf$ to $\sup$. The detailed proof is thus omitted. 
\end{proof}

\end{document}